\documentclass[11pt]{article}
\usepackage{ambecomen}
\usepackage{graphicx}
\usepackage{amsmath}
\usepackage{amsmath,scalefnt}
\usepackage{footnote}
\usepackage{fmtcount}
\usepackage{color}

\usepackage{threeparttable}
\usepackage{chemarr}

\newcommand{\norm}[1]{\left\Vert #1 \right\Vert}
\newcommand{\normpQ}[1]{\norm{#1}_{p,Q}}

\title{Remarks on contractions of reaction-diffusion PDE's\\
on weighted $L^2$ norms}
\author{{Zahra Aminzare}\\ \\
{\small Department of Mathematics, Rutgers University,}\\
{\small Piscataway, NJ 08854-8019 USA}}

\begin{document}
\maketitle

\section{Introduction}
The study of reaction-diffusion partial differential equations (PDE's) of the
form
\begin{equation} \label{re-di}
\displaystyle\frac{\partial u}{\partial t}(\omega, t)=F(u(\omega, t), t)+D\Delta u(\omega, t),
\end{equation}
where $\Delta$ denotes a diffusion operator, is central to many biological
applications, in fields ranging from pattern formation in development to
ecology.  One of the central topics of research in this context is the
question of how the stability of solutions of the PDE relates to stability of
solutions of the underlying ordinary differential equation (ODE)
$\frac{d x}{d t}(t)=F(x(t), t)$.
This paper shows that when solutions of this ODE have a certain contraction
property, namely $\mu_{2,P}(J_F(u, t))<0$ uniformly on $u$ and $t$, where
$\mu_{2,P}$ is a logarithmic norm (matrix measure) associated to a $P$-weighted
$L^2$ norm, the associated PDE, subject to no-flux (Neumann) boundary
conditions, enjoys a similar property, if $P^2D+DP^2>0$.  This result complements a similar
result shown in \cite{aminzare-sontag} which, while allowing norms $L^p$
with $p$ not necessarily equal to $2$, had the restriction that it only
applied to diagonal matrices $P$.  Here, $P$ is allowed to be an arbitrary
positive definite symmetric matrix.  The paper also discusses an example of
biological interest, as well as examples that illustrate when the results in
\cite{aminzare-sontag} apply but the current result does not.

\section{Main Result}
In this section, we study the reaction diffusion PDE (\ref{re-di}), subject to a Neumann boundary condition:
\begin{equation} \label{i-c}
\nabla u_i\cdot\mathbf{n}(\xi,t)=0\quad \forall\xi\in\partial\Omega,\;\;\forall t\in[0,\infty).
\end{equation} 

\begin{assumption}\label{as-pde}
In $(\ref{re-di})-(\ref{i-c})$ we assume:
\begin{itemize}
\item $\Omega$ is a bounded domain in $\r^m$ with smooth boundary $\partial\Omega$ and outward normal $\mathbf{n}$. 
\item $F\colon V\times [0, \infty)\to\r^n$ is a (globally) Lipschitz and twice continuously differentiable vector field with respect to $x$, and continuous with repect to $t$, with components $F_i$:
$$F(x, t)=(F_1(x, t), \cdots, F_n(x, t))^T$$ for some functions $F_i\colon V\times [0, \infty)\to\r$, where $V$ is a convex subset of $\r^n$.
\item $D=\diag(d_1, \cdots, d_n)$, with $d_i>0$, is called the diffusion matrix. 
\end{itemize}
\end{assumption}

\begin{definition}
By a solution of the PDE
\begin{equation*}
\displaystyle\frac{\partial u}{\partial t}(\omega, t)=F(u(\omega, t), t)+D\Delta u(\omega, t),
\end{equation*}
\begin{equation*} 
\nabla u_i\cdot\mathbf{n}(\xi,t)=0\quad \forall\xi\in\partial\Omega,\;\;\forall t\in[0,\infty).
\end{equation*}
on an interval $[0, T)$, where $0<T\leq\infty$, we mean a function $u=(u_1, \cdots, u_n)^T$, with $u\colon \displaystyle\bar\Omega\times [0,T)\to V$, such that:
\begin{enumerate}
 \item for each $\omega\in\bar\Omega$, $u(\omega, \cdot)$ is continuously differentiable;
 \item for each $t\in[0,T)$, $u(\cdot, t)$ is in $\mathbf{Y}$, where
 \[
\mathbf{Y}=\displaystyle \left\{v\colon\bar\Omega\to V\mid\;v=(v_1,\cdots, v_n),\quad v_i\in C^2_{\r}\left(\bar\Omega\right),\quad\frac{\partial v_i}{\partial\mathbf{n}}(\xi)=0,\; \forall\xi\in\partial\Omega\;\;\forall i\right\},
\]
where $C^2_{\r}\left(\bar\Omega\right)$ is the set of twice continuously differentiable functions $\bar\Omega\to \r$.
\item for each $\omega\in\bar\Omega$, and each $t\in [0,T)$, $u$ satisfies the above PDE.
\end{enumerate}
\end{definition}

Theorems on existence and uniqueness for PDE's such as $(\ref{re-di})-(\ref{i-c})$ can be found in standard references, e.g. \cite{Smith, Cantrell}.

For any invertible matrix $Q$, and any $1\leq p\leq\infty$, and continuous $u\colon\Omega\to\r^n$, we also denote $\|u\|_{p,Q}=\|Qu\|_p$, where $(Qu)(\omega)=Qu(\omega)$ and $\|\cdot\|_p$ now indicates the norm in $L^p(\Omega, \r^n)$. 

\begin{definition}
Let $(X,\|\cdot\|_X)$ be a finite dimensional normed vector space over $\r$ or $\cp$. The space $\mathcal{L}(X, X)$ of linear transformations $A\colon X \to X$ is also a normed vector space with the induced operator norm $$\|A\|_{X\to X}=\displaystyle\sup_{\|x\|_X=1}\|Ax\|_X.$$ The logarithmic norm $\mu_X(\cdot)$ induced by $\|\cdot\|_X$ is defined as the directional derivative of the matrix norm, that is,
\[
\mu_X(A)=\displaystyle\lim_{h\to 0^+}\frac{1}{h}\left(\|I+hA\|_{X\to X}-1\right),
\]
where $I$ is the identity operator on $X$. 
\end{definition}

In \cite{aminzare-sontag}, we proved the following lemma:
\begin{lemma}\label{contracting}
Consider the PDE system $(\ref{re-di})-(\ref{i-c})$, and suppose Assumption \ref{as-pde} holds. For some $1\leq p\leq\infty$, and a positive diagonal matrix $Q$, let
\[\mu\;:=\;\displaystyle\sup_{(x,t)\in V\times[0, \infty)}\mu_{p,Q}(J_F(x,t)).\] 
Then for any two solutions $u$ and $v$ of $(\ref{re-di})-(\ref{i-c})$, 
 \[\|u(\cdot, t)-v(\cdot, t)\|_{p,Q}\leq e^{\mu t}\|u(\cdot, 0)-v(\cdot, 0)\|_{p,Q}.\]
\end{lemma}

Before stating the main theorem of this section, we first prove the following:

\begin{lemma}\label{Lyapanov-inequality}
Suppose that $P$ is a positive definite, symmetric matrix and $A$ is an arbitrary matrix.
\begin{enumerate}
\item If $\mu_{2,P}(A)= \mu$, then $QA+A^TQ\leq 2\mu Q$, where $Q=P^2$. 
\item If for some $Q=Q^T>0$, $QA+A^TQ\leq 2\mu Q$, then there exists $P=P^T>0$ such that $P^2=Q$ and $\mu_{2,P}(A)\leq \mu$.
\end{enumerate}
\end{lemma}

\begin{proof}
First suppose $\mu_{2,P}(A)= \mu$. By definition of $\mu$:
\[\displaystyle\frac{1}{2}\left(PAP^{-1}+\left(PAP^{-1}\right)^T\right)\leq \mu I.\]
Since $P$ is symmetric, so is $P^{-1}$, so
\[PAP^{-1}+P^{-1}A^TP\leq 2\mu I.\]
Now multiplying the last inequality by $P$ on the right and the left, we get:
\[P^2A+A^TP^2\leq 2\mu P^2.\]
This proves $1$.
Now assume that for some $Q=Q^T>0$, $QA+A^TQ\leq 2\mu Q$. Since $Q>0$, there exists $P>0$ such that $P^TP=Q$; moreover, because $Q$ is symmetric, so is $P$. Hence we have:
\[P^2A+A^TP^2\leq 2\mu P^2.\]
Multiplying the last inequality by $P$ from right and by $P^{-1}$ from left, we conclude $2$. 
\end{proof}

\begin{remark}\label{Lyapanov-inequality2}
Observe that for $Q>0$,  
\begin{enumerate}
\item\[QA+A^TQ\leq \mu Q\quad\Rightarrow\quad QA+A^TQ\;\leq\; \beta I,\] where $\beta=\mu\lambda$ and $\lambda$ is the smallest eigenvalue of $Q$.
\item \[QA+A^TQ\;\leq\; \beta I\quad\Rightarrow\quad QA+A^TQ\;\leq\; \gamma Q,\] where $\gamma=\displaystyle\frac{\beta}{\lambda'}$ and $\lambda'$ is the largest eigenvalue of $Q$.

\end{enumerate}
\end{remark}

\begin{theorem}\label{cont-weighted-L2}
Consider the reaction diffusion system $(\ref{re-di})-(\ref{i-c})$ and suppose Assumption \ref{as-pde} holds. Let 
\[\mu:=\displaystyle\sup_{(x, t)\in V\times [0, \infty)}\mu_{2,P}(J_F(x, t)),\] 
for a positive symmetric (not necessarily diagonal) matrix $P$. In addition assume that $QD+DQ>0$, where $Q=P^2.$ Then for any two solutions, namely $u$ and $v$, of $(\ref{re-di})-(\ref{i-c})$, we have:
\begin{equation}\label{contraction}
 \|u(\cdot, t)-v(\cdot, t)\|_{2,P}\leq e^{\mu t}\|u(\cdot, 0)-v(\cdot, 0)\|_{2,P}.
\end{equation}
\end{theorem}

\begin{proof} \footnote[1]{The techniques of the proof are similar to the proof of Theorem $1$, \cite{Arcak}}
By Lemma \ref{Lyapanov-inequality}, 
\begin{equation}\label{Lyap-ineq}
QJ_F+J_F^TQ\leq2\mu Q,
\end{equation}
 where $Q=P^2.$ Let $u$, and $v$ be two solutions of $(\ref{re-di})-(\ref{i-c})$ and let $w:=u-v$. Define 
\[\Phi(w)\;:=\;\displaystyle\frac{1}{2}(w, Qw)=\displaystyle\frac{1}{2}\|Pw\|_2^2,\]
where $(x,y):=\displaystyle\int_{\Omega}x^Ty$. Since $u$, and $v$ satisfy $(\ref{re-di})$, we have:
\begin{equation}\label{Phi-dot}
\displaystyle\frac{d\Phi}{dt}(w)=\left(w, Q(F(u, t)-F(v, t))\right)+\left(w, QD\Delta w\right).
\end{equation}
Since $QD+DQ$ is a positive matrix, there exists a positive, symmetric matrix $M$ such that $QD+DQ=2M^TM$.
\begin{equation}\label{eq1:proof of theorem}
(w, (QD+DQ)\Delta w)\;=\;2(w, M^TM\Delta w)\;=\;2(Mw,\Delta Mw)\;=\;-2(\nabla(Mw), \nabla(Mw)).
\end{equation}
The last equality holds by the Mean Value Theorem and the fact that $Mw$ satisfies the Neumann boundary conditions.   

By Green's identity and using $DQ=(QD)^T$, we have:
\begin{equation}\label{eq2:proof of theorem}
(w, (QD-DQ)\Delta w)\;=\;(w, \Delta QDw)-(QDw, \Delta w)\;=\;0.
\end{equation}
Equations (\ref{eq1:proof of theorem}) and (\ref{eq2:proof of theorem}) imply:
\[(w, QD\Delta w)\;=\;-(\nabla(Mw), \nabla(Mw))\;\leq\;0.\]
Now by Mean Value Theorem for integrals, and using (\ref{Lyap-ineq}), we rewrite the first term of the right hand side of (\ref{Phi-dot}) as follows:
 \[
 \begin{array}{lcl}
 \left(w, Q(F(u, t)-F(v, t))\right)&=&\displaystyle\int_{\Omega}w^T(\omega, t)Q(F(u(\omega, t), t)-F(v(\omega, t), t))\;d\omega\\
 &=&\displaystyle\int_{\Omega}w^T(\omega, t)Q\displaystyle\int_0^1J_F(v(\omega, t)+sw(\omega, t), t)\cdot w(\omega, t)\;ds\;d\omega\\
  &\leq&\mu\displaystyle\int_0^1\;ds\displaystyle\int_{\Omega}w^TQw\;d\omega\\
  &=&\mu\displaystyle\int_{\Omega}w^TQw\;d\omega\\
  &=&2\mu\Phi(w).
 \end{array}
 \]
 Therefore \[\displaystyle\frac{d\Phi}{dt}(w)\;\leq\; 2\mu\Phi(w).\]
 This last inequality implies (\ref{contraction}).
\end{proof}
\begin{corollary}\label{cor-cont-weighted-L2}
In Theorem \ref{cont-weighted-L2}, if $\mu<0$, then $(\ref{re-di})-(\ref{i-c})$ is contracting, meaning that solutions converge (exponentially) to each other, as $t\to+\infty$ in the weighted $L_{2,P}$ norm.
\end{corollary}

\textbf{Example $1$.}
In \cite{aminzare-sontag} we studied the following system: 
\[
\begin{array}{lcl}\label{example}
x_t=z-\delta x+k_1y-k_2(S_Y-y)x+d_1\Delta x\\
y_t=-k_1y+k_2(S_Y-y)x+d_2\Delta y,
\end{array}
\]
where $(x(t), y(t))\in V=[0, \infty)\times[0, S_Y]$ for all $t\geq0$ ($V$ is convex), and $S_Y$, $k_1$, $k_2$, $\delta$, $d_1$, and $d_2$ are arbitrary positive constants.
 
 In \cite{Russo}, it has been shown that for $p=1$, there exists a positive, diagonal matrix $Q$ such that for all $(x,y)\in V$, $\mu_{1,Q}(J_F(x,y))<0$; and then by Lemma \ref{contracting} one can conclude that the system is contractive. 
 
 We showed that for any positive, diagonal matrix $Q$ and any $p>1$, there exists $(x,y)\in V$ such that $\mu_{p,Q}(J_F(x,y))\geq 0$, where 
\[F=(z-\delta x+k_1y-k_2(S_Y-y)x, -k_1y+k_2(S_Y-y)x)^T,\] and 
\[J_F=\displaystyle\left(\begin{array}{cc}-\delta-a & b\\a & -b\end{array}\right),\]
 with $a=k_2(S_Y-y)\in[0,k_2S_Y]$ and $b=k_1+k_2x\in[k_1,\infty)$.

Now we show that there exists some positive, symmetric (but non-diagonal) matrix $P$ such that for all $(x,y)\in V$, $\mu_{2,P}J_F(x,y)<0$. Then for those $d_1$ and $d_2$ that satisfy $P^2D+DP^2>0$, where $D=\diag(d_1, d_2)$, Corollary \ref{cor-cont-weighted-L2} concludes that the system is contractive. 

{\em Claim.} Let $Q= \mbox{%
$\begin{bmatrix}1 & 1\\1 & q\end{bmatrix}$}$, where $q>1+\displaystyle\frac{\delta}{4k_1}$. Then $QJ_F+(QJ_F)^T<0$. 

Note that $Q$ is symmetric and positive (because $q>1$).

{\em Proof of Claim.} We first compute
\[QJ_F= \mbox{%
$\begin{bmatrix}1 & 1\\1 & q\end{bmatrix}$
$\begin{bmatrix}-\delta-a & b\\a & -b\end{bmatrix}$ =
$\begin{bmatrix}-\delta & 0\\-\delta+(q-1)a & -b(q-1)\end{bmatrix}$}.\]
So 
\[QJ_F+(J_FQ)^T= \mbox{%
$\begin{bmatrix}-2\delta &-\delta+(q-1)a \\-\delta+(q-1)a & -2b(q-1)\end{bmatrix}$}.\]
To show $QJ_F+J_F^TQ<0$, we show that $\det \left(QJ_F(x,y)+J_F^T(x,y)Q\right)>0$ for all $(x,y)\in V$:
\[\det \left(QJ_F+J_F^TQ\right)=4\delta b(q-1)-\left(-\delta+(q-1)a\right)^2.\]

Note that for any $q>1$, $f(a):=\left(-\delta+(q-1)a\right)^2\leq \delta^2$ on $[0, k_2S_Y]$, and $g(b):=4\delta b(q-1)\geq 4\delta k_1(q-1)$ on $[k_1,\infty]$. So to have $\det>0$, it's enough to have $4\delta k_1(q-1)-\delta^2>0$, i.e. $q-1>\displaystyle\frac{\delta^2}{4\delta k_1}$, i.e. $q>1+\displaystyle\frac{\delta}{4 k_1}$. \qed

Now by Remark \ref{Lyapanov-inequality2} and Lemma \ref{Lyapanov-inequality}, for $P=\sqrt Q$, $\mu_{2,P}(J_F(x,y))<0$, for all $(x,y)\in V.$

\textbf{Example $2$.}  We now provide an example of a class of reaction-diffusion systems $x_t=F(x)+D\Delta x$, with $x\in V$ ($V$ convex), such that for some positive definite, diagonal matrix $Q$, and for all $x\in V$, $\mu_{1,Q}(J_F(x))<0$ (and hence by Lemma \ref{contracting}, these systems are contractive), yet for these systems, $\mu_{2,P}(J_F(x))\nless0$, even for any positive definite, symmetric (not necessarily diagonal) matrix $P$. Consider two variable systems of the following type 
\begin{equation}\label{mu1-vs-mu2-1}
x_t=-f_1(x)+g_1(y)+d_1\Delta x\\
\end{equation}
\begin{equation}\label{mu1-vs-mu2-2}
y_t=f_2(x)-g_2(y)+d_2\Delta y,
\end{equation}
where $d_1$, $d_2$ are positive constants and $(x,y)\in V=[0,\infty)\times[0,\infty)$. The functions $f_i$ and $g_i$ take non-negative values. Systems of this form models a case where $x$ decays according to $f_1$, $y$ decays according to $g_2$, and there is a positive feedback from $y$ to $x$ ($g_1$) and a positive feedback from $x$ to
$y$ ($f_2$). 

\begin{lemma}\label{conditions-mu1-vs-mu2}
In system (\ref{mu1-vs-mu2-1})-(\ref{mu1-vs-mu2-2}), let $J$ be the Jacobian matrix of 
\[\left(-f_1(x)+g_1(y), f_2(x)-g_2(y)\right)^T.\] In addition, assume that the following conditions hold for some $\lambda>0$, and $\mu>0$ and all $(x,y)\in V$:
\begin{enumerate}
\item $-f_1'(x)+\lambda \abs{f_2'(x)}<-\mu<0$;
\item$ -g_2'(y)+\displaystyle\frac{1}{\lambda} \abs{g_1'(y)}<-\mu<0$;
\item for any $p_0\in\r$
\begin{equation*}\label{rate-infty}
\displaystyle\lim_{y\to\infty}\displaystyle\frac{\left(g_1'(y)-p_0g_2'(y)\right)^2}{g_2'(y)}=\infty.
\end{equation*}
\end{enumerate}

Then 
\begin{enumerate}
\item for every $(x,y)\in V$, $\mu_{1,Q}(J(x,y))<0$, where $Q=\diag(1,\lambda)$; and
\item for each positive definite, symmetric matrix $P$, there exists some $(x,y)\in V$, such that $\mu_{2,P}(J(x,y))\geq0$.
\end{enumerate}
\end{lemma}
  
\begin{proof} 
The proof of $\mu_{1,Q}(J(x,y))<0$ is straightforward from the definition of $\mu_{1,Q}$ and conditions $1$ and $2$. Now we show that for any positive matrix $P= \mbox{%
$\begin{bmatrix}{p_1} & p \\p & {p_2}\end{bmatrix}$}$, there exists some $(x_0, y_0)\in V$ such that $\mu_{2,P}(J(x_0,y_0))\geq0$. By Lemma \ref{Lyapanov-inequality}, it's enough to show that for some $(x_0, y_0)\in V$, $PJ(x_0,y_0)+J^T(x_0,y_0)P\nless0$. We compute:
\[ 
\mbox{%
$PJ=\begin{bmatrix}
p_1&p\\p& p_2
  \end{bmatrix}
  \begin{bmatrix}
-f_1'(x) & g_1'(y) \\ f_2'(x) & -g_2'(y)
  \end{bmatrix}=\displaystyle\begin{bmatrix}
-p_1f_1'(x)+pf_2'(x)&p_1g_1'(y)-pg_2'(y)\\-pf_1'(x)+p_2f_2'(x)& pg_1'(y)-p_2g_2'(y)
  \end{bmatrix}
  $}. 
  \]
Therefore, 
\[ 
\mbox{%
$PJ+(PJ)^T=\displaystyle\begin{bmatrix}
2\left(-p_1f_1'(x)+pf_2'(x)\right)&p_1g_1'(y)-pg_2'(y)-pf_1'(x)+p_2f_2'(x)\\p_1g_1'(y)-pg_2'(y)-pf_1'(x)+p_2f_2'(x)& 2\left(pg_1'(y)-p_2g_2'(y)\right)
  \end{bmatrix}
  $}. 
  \]
Now fix $x_0\in [0,\infty)$ and let 
\[A\;:=\;2\left(-p_1f_1'(x_0)+pf_2'(x_0)\right)\quad \mbox{and}\quad B\;:=\;-pf_1'(x_0)+p_2f_2'(x_0).\]
 Then 
 \begin{equation}\label{det}
 \det\left(PJ+(PJ)^T\right)=2A\left(pg_1'(y)-p_2g_2'(y)\right)-\left(p_1g_1'(y)-pg_2'(y)+B\right)^2.
 \end{equation}
 We will show that $\det<0$. 
 Dividing both sides of (\ref{det}) by $p_1^2g_2'(y)$, we get: 
  \[
   \begin{array}{rcl}
\displaystyle\frac{\det\left(PJ+(PJ)^T\right)}{p_1^2g_2'(y)}&=&\displaystyle\frac{2A\left(pg_1'(y)-p_2g_2'(y)\right)}{p_1^2g_2'(y)}
-\displaystyle\frac{\left(g_1'(y)-p_0g_2'(y)+B'\right)^2}{g_2'(y)}\\
&=&A'p\displaystyle\frac{g_1'(y)}{g_2'(y)}-A'p_2\\
&-&\displaystyle\frac{\left(g_1'(y)-p_0g_2'(y)\right)^2}{g_2'(y)}-2B'\displaystyle\frac{g_1'(y)}{g_2'(y)}+2B'p_0-\displaystyle\frac{B'^2}{g_2'(y)}\\
 \end{array}
 \]
where $p_0=\displaystyle\frac{p}{p_1}$, $A'=\displaystyle\frac{2A}{p_1^2}$, and $B'=\displaystyle\frac{B}{p_1}$. 

(Note that $p_1^2g_2'(y)>0$ because by condition $2$, $g_2'\geq\mu>0$, and $P>0$ implies $p_1\neq0$.)

By condition $2$, $0\leq\displaystyle\frac{g_1'(y)}{g_2'(y)}\leq\lambda<\infty$ for all $y$. Now using condition $3$, we can find $y$ large enough such that $\det<0.$

 Since $\det\left(PJ(x_0,y_0)+(PJ(x_0,y_0))^T\right)<0$ for some $(x_0,y_0)\in V$, the matrix $PJ+(PJ)^T$ has one positive eigenvalue. Therefore $PJ+(PJ)^T\nless0$. 
\end{proof}

 As a concrete example, take the following system 
\[
\begin{array}{lcl}
x_t=-x+y^{2+\epsilon}+d_1\Delta x\\
y_t=\delta x-(y^3+y^{2+\epsilon}+dy)+d_2\Delta y,
\end{array}
\]
where $0<\delta<1$, $0<\epsilon\ll1$, $d$, $d_1$, and $d_2$ are positive constants and $(x,y)\in V=[0, \infty)\times[0, \infty)$.

In this example we show that, the system is contractive in a weighted $L^1$ norm; while for any positive, symmetric matrix $P$, and some $(x,y)\in V$, $\mu_{2,P}J_F(x,y)\nless0$. To this end, we verify the conditions of Lemma \ref{conditions-mu1-vs-mu2}.

For any $(x,y)\in V$, we take in Lemma \ref{conditions-mu1-vs-mu2}, $\lambda=1$, and any $\mu\in(0,\min\{d, 1-\delta\})$:
\begin{enumerate}
\item $-1+\delta<0$, because $0<\delta<1$.
\item $ -\left(3y^2+(2+\epsilon)y^{1+\epsilon}+d\right)+(2+\epsilon)y^{1+\epsilon}=-3y^2-d\leq-d<0.$
\item For any $p_0\in \r$, 
\[\displaystyle\lim_{y\to\infty}\displaystyle\frac{\left((2+\epsilon)y^{1+\epsilon}-p_0\left(3y^2+(2+\epsilon)y^{1+\epsilon}+d\right)\right)^2}{3y^2+(2+\epsilon)y^{1+\epsilon}+d}=\infty.\]
\end{enumerate}
So the conditions in Lemma \ref{conditions-mu1-vs-mu2} are verified. \qed

\section{Diffusive interconnection of ODEs}
In this section, we derive a result analogous to that for PDE's for a network of identical ODE models which are diffusively interconnected. We study systems of ODE's as follows:
\begin{equation}\label{discrete}
\dot{u}(t)=\tilde{F}(u(t))-(L\otimes D)u(t).
\end{equation}
\begin{assumption}\label{as-ode}
In $(\ref{discrete})$, we assume:
\begin{itemize}
\item For a fixed convex subset of $\r^n$, say $V$, $\tilde{F}\colon V^{N}\to\r^{nN}$ is a function of the form:
\[\tilde{F}(u)=\left(F(u^1)^T, \cdots, F(u^N)^T\right)^T,\]
 where $u=\left({u^1}^T,\cdots, {u^N}^T\right)^T$, with $u^i\in V$ for each $i$, and $F\colon V\to\r^n$ is a (globally) Lipschitz function. 
\item For any $u\in V^N$ we define $\normpQ{u}$ as follows:
\[\normpQ{u}=\left\|\left(\|Qu^1\|_p, \cdots, \|Qu^N\|_p\right)^T\right\|_p,\]
where $Q$ is a symmetric and positive definite matrix and $1\leq p\leq\infty$.

With a slight abuse of notation, we use the same symbol for a norm in $\r^n$: \[\|x\|_{p,Q}:=\|Qx\|_p.\]

\item $u\colon [0,\infty)\to V^N$ is a continuously differentiable function.
\item $D=\diag(d_1,\cdots, d_n)$ with $d_i>0$, which we call the diffusion matrix.
\item $L\in\r^{N\times N}$ is a symmetric matrix and $L\mathbf{1}=0$, where $\mathbf{1}=(1,\cdots, 1)^T\in\r^{N}$. 
We think of $L$ as the Laplacian of a graph that describes the interconnections among component subsystems. 

\end{itemize}
\end{assumption}

In \cite{aminzare-sontag}, we proved the following lemma:
\begin{lemma}\label{contracting}
Consider the ODE system (\ref{discrete}), and suppose Assumption \ref{as-ode} holds. For some $1\leq p\leq\infty$, and a positive diagonal matrix $Q$, let
\[\mu\;:=\;\displaystyle\sup_{(x,t)\in V\times[0, \infty)}\mu_{p,Q}(J_F(x,t)).\]
($\mu_{p,Q}$ is the logarithmic norm induce by the norm $\|\cdot\|_{p,Q}$ on $\r^n$ defined by $\|x\|_{p,Q}:=\|Qx\|_p$.)
 
Then for any two solutions $u$ and $v$ of $(\ref{discrete})$, 
 \[\|u(t)-v(t)\|_{p,Q}\leq e^{\mu t}\|u(0)-v(0)\|_{p,Q}.\]
\end{lemma}

In this section we generalize the above result for $p=2$ and any symmetric, positive definite (not necessarily diagonal) matrix $P$ such that $P^2D+DP^2>0$. 
\begin{theorem}\label{discrete-weighted-L2}
Consider the ODE system (\ref{discrete}) and suppose Assumption \ref{as-ode} holds. Let 
\[\mu:=\displaystyle\sup_{(x, t)\in V\times [0, \infty)}\mu_{2,P}(J_F(x, t)),\] 
for a positive symmetric (not necessarily diagonal) matrix $P$. In addition assume that $QD+DQ>0$, where $Q=P^2.$ Then for any two solutions, namely $u$ and $v$, of (\ref{discrete}), we have:
\begin{equation}\label{contraction_discrete}
 \|u(t)-v(t)\|_{2,P}\leq e^{\mu t}\|u(0)-v(0)\|_{2,P}.
\end{equation}
\end{theorem}

Before proving the theorem, we recall that if $A=(a_{ij})$ is an $m\times n$ matrix and $B=(b_{ij})$ is a $p\times q$ matrix, then the Kronecker product, denoted by $A\otimes B$, is the $mp\times nq$ block matrix defined as follows:
\[A\otimes B\;:=\;\mbox{%
$\displaystyle\begin{bmatrix}a_{11}B & \ldots &a_{1n}B \\\vdots&\ddots&\vdots\\a_{m1}B & \ldots & a_{mn}B\end{bmatrix}$},\]
where $a_{ij}B$ denote the following $p\times q$ matrix:
\[a_{ij}B\;:=\;\mbox{%
$\displaystyle\begin{bmatrix}a_{ij}b_{11} & \ldots &a_{ij}b_{1q} \\\vdots&\ddots&\vdots\\a_{ij}b_{p1} & \ldots & a_{ij}b_{pq}\end{bmatrix}$}.\]

The following are some properties of Kronecker product: 
\begin{enumerate}
\item $(A\otimes B)(C\otimes D)=(AC)\otimes (BD)$;
\item $(A\otimes B)^T=\;A^T\otimes B^T.$
\item Suppose that $A$ and $B$ are square matrices of size $n$ and $m$ respectively. Let $\lambda_1, \cdots, \lambda_n$ be the eigenvalues of $A$ and $\mu_1, \cdots, \mu_m$ be those of $B$ (listed according to multiplicity). Then the eigenvalues of $A\otimes B$ are $\lambda_i \mu_j$ for $i=1, \cdots, n$, and $j=1, \cdots, m$. 
\end{enumerate}

\newtheorem*{proofode}{Proof of Theorem \ref{discrete-weighted-L2}}
\begin{proofode} \footnote[2]{The techniques of the proof are similar to Theorem $4$, \cite{Arcak}}
By Lemma \ref{Lyapanov-inequality}, 
\begin{equation}\label{Lyap-ineq_discrete}
QJ_F+J_F^TQ\leq 2\mu Q,
\end{equation}
 where $Q=P^2.$ Let $u$, and $v$ be two solutions of (\ref{discrete}) and let $w:=u-v$. Define 
\[\Phi(w)\;:=\;\displaystyle\frac{1}{2}w^T(I_N\otimes Q)w.\]
Note that indeed 
\begin{subequations}
\begin{align*}\Phi(w) = \displaystyle\frac{1}{2}\sum_{i=1}^N{w^i}^TQ w^i 
=\displaystyle\frac{1}{2}\sum_{i=1}^N(Pw^i)^T (Pw^i)
=\displaystyle\frac{1}{2}\sum_{i=1}^N\|Pw^i\|_2^2
=\displaystyle\frac{1}{2}\|w\|_{2, P}^2
\end{align*}
\end{subequations}

Since $u$, and $v$ satisfy (\ref{discrete}), using the first property of Kronecker product listed above, we have:
\begin{subequations}\label{Phi-dot_discrete}
\begin{align*}
\displaystyle\frac{d\Phi}{dt}(w) &= w^T(I_N\otimes Q)(\tilde{F}(u, t)-\tilde{F}(v, t))-w^T(I_N\otimes Q)(L\otimes D) w\\
&= w^T(I_N\otimes Q)(\tilde{F}(u, t)-\tilde{F}(v, t))-w^T(L\otimes QD) w.
\end{align*}
\end{subequations}
Since $QD+DQ$ is a positive matrix, there exists a positive, symmetric matrix $M$ such that $QD+DQ=2M^TM$. Then using 
\[L\otimes (M^TM)\;=\;(I_N\otimes M^T)(L\otimes I_n)(I_N\otimes M),\] we get:
\begin{subequations}\label{eq1:proof of theorem_discrete}
\begin{align}
w^TL\otimes(QD+DQ) w &= 2w^T(I_N\otimes M^T)(L\otimes I_n)(I_N\otimes M)w\\
&= 2\left((I_N\otimes M)w\right)^T(L\otimes I_n) ((I_N\otimes M)w)\;\geq\;0.
\end{align}
\end{subequations}
The last inequality holds because all the eigenvalues of $L$, and therefore all the eigenvalues of $L\otimes I_n$, are non-negative, by the third property of Kronecker product listed above. Now because $L=L^T$, $DQ=(QD)^T$, using the second property of Kronecker product listed above, we get:
\begin{subequations}\label{eq2:proof of theorem_discrete}
\begin{align}
w^TL\otimes(QD-DQ) w &= w^T\left[(L\otimes QD)w\right]- w^T\left(L^T\otimes (QD)^T\right) w\\
&= w^T\left[(L\otimes QD)w\right]- w^T\left(L\otimes QD\right)^T w\\
&= w^T\left[(L\otimes QD)w\right]-\left[(L\otimes QD)w\right]^T w=0.
\end{align}
\end{subequations}
Considering the sum of Equations (\ref{eq1:proof of theorem_discrete}) and (\ref{eq2:proof of theorem_discrete}), we get:
\[-w^T(L\otimes QD) w\;\leq\; 0.\]
Now by Mean Value Theorem for integrals, and using (\ref{Lyap-ineq_discrete}), we rewrite the first term of the right hand side of (\ref{Phi-dot_discrete}) as follows:
 \[
 \begin{array}{lcl}
 w^T(I_N\otimes Q)(\tilde{F}(u, t)-\tilde{F}(v, t))&=&\displaystyle\sum_{i=1}^N{w^i}^TQ(F(u^i, t)-F(v^i, t)) w^i\\
 &=&\displaystyle\sum_{i=1}^N w_i^TQ\displaystyle\int_0^1J_F(v^i+sw^i, t)w^i\;ds\\
  &\leq&\mu\displaystyle\int_0^1\;ds\displaystyle\sum_{i=1}^N{w^i}^TQw^i\\
  &=&\mu w^T(I_N\otimes Q) w\\
  &=&2\mu\Phi(w).
 \end{array}
 \]
 Therefore \[\displaystyle\frac{d\Phi}{dt}(w)\;\leq\; 2\mu\Phi(w).\]
 This last inequality implies (\ref{contraction_discrete}).\qed
\end{proofode}

\end{document}